\documentclass[12pt,draftcls,journal,onecolumn]{IEEEtran}

\usepackage{amssymb,amsthm, amsmath,latexsym}
\usepackage{graphicx}
\usepackage{mathrsfs}
\usepackage{amsfonts}
\usepackage{amssymb}
\usepackage{longtable}
\usepackage{amsmath}
\usepackage{setspace}
\usepackage{caption}
\usepackage[figuresright]{rotating}
\usepackage[misc]{ifsym}
\usepackage{bbm}
\usepackage{makecell}
\usepackage{arydshln}
\usepackage{supertabular}
\usepackage{booktabs}

\newtheorem{theorem}{Theorem}
\newtheorem{lemma}[theorem]{Lemma}
\newtheorem{remark}[theorem]{Remark}
\newtheorem{proposition}[theorem]{Proposition}

\newtheorem{example}[theorem]{Example}

\usepackage{blindtext}

\ifCLASSINFOpdf

\else

\fi

\begin{document}
\title{ Infinite families of $3$-designs from linear and nonlinear codes
\thanks{
}
}
\author{Shiyan Xiong, Xiaoqiang Wang$^{*}$, Dabin Zheng, Qinqin Ji}

\renewcommand{\thefootnote}{\empty}
\footnotetext{\thanks{*Corresponding author. }
\newline \indent Shiyan Xiong, Xiaoqiang Wang, and Dabin Zheng are with the Hubei Key Laboratory of Applied Mathematics, Faculty of Mathematics and Statistics, Hubei University, Wuhan 430062, China (E-mail:  xionshiyan@163.com;waxiqq@163.com;  dzheng@hubu.edu.cn).
\newline \indent Qinqin Ji is with School of Mathematics and Statistics, Hubei University of Education, Wuhan 430205, China (E-mail: qqinji@163.com).}

\maketitle

\begin{abstract}
The connection between coding theory and combinatorial $t$-designs is an important research topic at the intersection of coding theory and combinatorics.
Let $q=p^m$, where $p$ is an odd prime and $m\geq 2$. In this paper, we investigate a class of linear codes $\mathcal{C}$ over $\mathbb{F}_{q^2}$ and their connection with combinatorial $3$-designs.
By analyzing the relevant structural properties of $\mathcal{C}$ and $\mathcal{C}^{\perp}$, we show that the supports of the codewords of every nonzero weight in $\mathcal{C}$ and {the supports of the codewords of weight $4$ in $\mathcal{C}^{\perp}$} form $3$-designs.
We further investigate a class of nonlinear codes $\mathcal{C}_2$ associated with $\mathcal{C}$, and prove that the supports of the codewords of every nonzero Hamming weight in $\mathcal{C}_2$ also form $3$-designs. These results identify further classes of linear and nonlinear codes whose codewords support combinatorial $3$-designs. In particular, the nonlinear case provides additional examples of codes supporting $3$-designs, a topic for which relatively few results are currently available.
As applications, we construct from $\mathcal{C}^{\perp}$ an entanglement-assisted quantum error-correcting code with parameters $[[q+1,q-3,4;4]]_q$. We also prove that $\mathcal{C}$ is an all-symbol locally repairable code with locality $3$. Furthermore, we show that the code $\mathcal{C}$ {meets the Singleton-type bound for locally repairable codes} and hence is optimal in some cases.
\end{abstract}



%
\IEEEpeerreviewmaketitle

\section{Introduction}\label{sec:introduction}

Let $q$ be a power of a prime $p$, and let $\mathbb{F}_q$ denote the
finite field with $q$ elements. An $[n,k,d]_q$ linear code
$\mathcal{C}$ over $\mathbb{F}_q$ is a $k$-dimensional subspace of
$\mathbb{F}_q^n$ with minimum Hamming distance $d$. Its dual code is
defined by
\[
\mathcal{C}^{\perp}
=
\left\{
\mathbf{a}\in\mathbb{F}_q^n:
\langle\mathbf{a},\mathbf{c}\rangle=0
\text{ for all }\mathbf{c}\in\mathcal{C}
\right\},
\]
where $\langle\cdot,\cdot\rangle$ denotes the standard inner product.
The dimension of $\mathcal{C}^{\perp}$ is $n-k$. In contrast, a
nonlinear code is an arbitrary subset of $\mathbb{F}_q^n$ without
the linearity requirement. If $A_i$ denotes the number of codewords
of Hamming weight $i$ in a code, then
$(A_0,A_1,\ldots,A_n)$ is called the weight distribution of this code.

Combinatorial $t$-designs are closely related to coding theory.
Let $\mathcal{P}$ be a set of $n$ points and let $\mathcal{B}$ be a
collection of $w$-subsets of $\mathcal{P}$. The incidence structure
$\mathbb{D}=(\mathcal{P},\mathcal{B})$ is called a
$t$-$(n,w,\lambda)$ design if every $t$-subset of $\mathcal{P}$ is
contained in exactly $\lambda$ blocks of $\mathcal{B}$. A $t$-design
is called simple if no block is repeated, and it is called a Steiner
system $S(t,w,n)$ if $\lambda=1$. Let $b$ denote the number of blocks
in $\mathcal{B}$. The parameters of a $t$-$(n,w,\lambda)$ design
satisfy
\begin{equation}\label{eq:design}
\binom{n}{t}\lambda=\binom{w}{t}b.
\end{equation}

Let ${X=\{1,2,\ldots,n\}}$ be the set of coordinate positions of the
codewords of a code $\mathcal{C}$ of length $n$. For a codeword
$\mathbf{c}=(c_1,c_2,\ldots,c_n)\in\mathcal{C}$, the support of
$\mathbf{c}$ is defined by
$$
\operatorname{supp}(\mathbf{c})
=
\{i\in X:c_i\neq 0\}.
$$
Let $\mathcal{B}_w(\mathcal{C})$ denote the set of distinct supports
of all codewords of Hamming weight $w$ in $\mathcal{C}$. If the
incidence structure
$
(X,\mathcal{B}_w(\mathcal{C}))
$
forms a $t$-$(n,w,\lambda)$ design for some positive integers $t$ and
$\lambda$, where $1\leq w\leq n$, then we say that the codewords of
weight $w$ in $\mathcal{C}$ support a $t$-design, or simply that
$\mathcal{C}$ supports a $t$-design.
The classical Assmus--Mattson theorem
\cite{AssmusMattson} provides an important sufficient condition for
deriving $t$-designs from the weight distributions of a code and its
dual. Another useful approach is based on highly homogeneous or
transitive automorphism groups
\cite{LiuDingMesnagerTangTonchev2,DingTangTonchevPGL}.
Basic results on linear codes and combinatorial designs can be found in
\cite{HuffmanPless,DingTangBook}.

A number of infinite families of codes supporting $t$-designs have
been obtained in recent years. Ding and Tang
\cite{DingTangNMDS} constructed infinite families of NMDS codes arising
from BCH codes and showed that they support $2$-designs and
$3$-designs. Subsequently, they \cite{TangDing4Design} constructed an
infinite family of linear codes supporting $4$-designs, providing the
first infinite family of $4$-designs supported by linear codes obtained
in approximately 70 years. Yan and Zhou \cite{YanZhou} further
investigated the same codes and obtained additional $3$-designs and
$4$-designs. Xiang et al. \cite{XiangTangLiu} constructed a family of
antiprimitive cyclic codes supporting infinite families of
$3$-designs. Other related results on designs arising from linear
codes, cyclic codes, almost MDS codes, special functions, and special
polynomials can be found in
\cite{DingLiLi,DingTangCao,TangAPN,DingTangSpecialPolynomials,
DuWangFan,TangDingXiong,XiangTangLiu,YanZhou,XuCaoQu,
TonchevCodesDesigns,KhosrovshahiLaue}.

Beyond their combinatorial properties, structured codes also have
important applications in quantum error correction and distributed
storage. Entanglement-assisted quantum error-correcting codes
(EAQECCs) were introduced by Bowen \cite{Bowen} and further developed
by Brun et al. \cite{BrunDevetakHsieh}. By using pre-shared
entanglement, EAQECC constructions relax the self-orthogonality
restriction in conventional quantum-code constructions
\cite{WildeBrun}. Various constacyclic and negacyclic codes have
therefore been investigated for constructing quantum codes
\cite{ChenLingZhang,ZhuSunLi}. Locally repairable codes (LRCs), on the
other hand, are designed to recover an erased symbol by accessing only
a small number of other symbols. Fundamental bounds on LRCs were
established in
\cite{GopalanHuangSimitciYekhanin,CadambeMazumdar}, and several
families of optimal cyclic and constacyclic LRCs have been constructed
in \cite{ChenFangXiaFu,SunZhuWang,TanZhouYanParampalli}.

 Of particular relevance to the present work,
Wang, Tang, and Ding investigated a family of cyclic codes with
nonzeros
$\beta^{(p^s-1)/2}$ and $\beta^{(p^s+1)/2}$ over
$\mathbb{F}_q$, and proved that these
codes support infinite families of $3$-designs
\cite{WangTangDing}, where $q=p^m$, $p$ is an odd prime, $m,s$ are positive integers
with $m>s$, and $\beta$ is a primitive $(q+1)$-th root of unity.  Let
$
k_1=\frac{p^s-1}{2}
$ and $
k_2=\frac{p^s+1}{2},
$
the code considered by Wang et al. \cite{WangTangDing} can be represented
as
$
\mathcal{C}_1=
\left\{
(a,b,a^q,b^q)A:
a,b\in\mathbb{F}_{q^2}
\right\},
$
where {$n=q+1$ and}
\begin{equation}\label{eq:C070201}
A=
\begin{pmatrix}
\beta^{k_1}  & (\beta^2)^{k_1}  & \cdots & (\beta^n)^{k_1}\\
\beta^{k_2}  & (\beta^2)^{k_2}  & \cdots & (\beta^n)^{k_2}\\
\beta^{-k_1} & (\beta^2)^{-k_1} & \cdots & (\beta^n)^{-k_1}\\
\beta^{-k_2} & (\beta^2)^{-k_2} & \cdots & (\beta^n)^{-k_2}
\end{pmatrix}.
\end{equation}
Motivated by the work of \cite{WangTangDing}, we further investigate codes related to their
construction.
We consider the linear code
\begin{equation}\label{eq:C}
\mathcal{C}=
\left\{
(a,b,c,d)A:
a,b,c,d\in\mathbb{F}_{q^2}
\right\}
\end{equation}
and the associated nonlinear code
\begin{equation}\label{eq:C2}
\mathcal{C}_2=
\left\{
(a,b,-a^q,-b^q)A:
a,b\in\mathbb{F}_{q^2}
\right\},
\end{equation}
where $A$ is given in (\ref{eq:C070201}).

In this paper, we show that these three codes satisfy
\[
\mathcal{C}=
\left\{
\mathbf{c}_1+\mathbf{c}_2:
\mathbf{c}_1\in \mathcal{C}_1,\,
\mathbf{c}_2\in \mathcal{C}_2
\right\}.
\]
Since the properties of $\mathcal{C}_1$ have already been studied in \cite{WangTangDing}, we mainly discuss the properties of $\mathcal{C}$ and $\mathcal{C}_2$.
Firstly, we determine the parameters and weight distribution of $\mathcal{C}$, and it is proved that the codewords of each nonzero Hamming weight in $\mathcal{C}$ support a $3$-design. By further investigating the weight distribution of the dual code $\mathcal{C}^\perp$, we show that the codewords of Hamming weight $4$ in $\mathcal{C}^\perp$ also support a family of $3$-designs. Secondly, as applications of the linear code $\mathcal{C}$, an entanglement-assisted quantum error-correcting code with parameters $[[q+1,q-3,4;4]]_q$ is constructed from $\mathcal{C}^\perp$. In addition, $\mathcal{C}$ is shown to be an all-symbol locally repairable code with locality $3$. When $p=3$ and $\gcd(m,s)=1$, the code $\mathcal{C}$ meets the LRC Singleton-type bound and hence forms a class of optimal locally repairable codes. Finally, the relationship between the supports of codewords in $\mathcal{C}$ and $\mathcal{C}_2$ is established, and the parameters of nonlinear code $\mathcal{C}_2$ are determined. Moreover, it is shown that the codewords of each nonzero Hamming weight in $\mathcal{C}_2$ support the corresponding $3$-designs.

The remainder of this paper is organized as follows. Section I
introduces the background, related work, and the main motivation of
this paper. Section II presents the basic concepts and preliminary
results needed in the subsequent sections. In Section III, we mainly investigate the $3$-designs supported by
$\mathcal{C}$, and further study some applications of $\mathcal{C}$. In Section IV, we investigate the properties of the nonlinear code $\mathcal{C}_2$. Section V concludes the paper.

\section{Preliminaries}\label{sec:preliminaries}

Let $\mathbb{F}_{q}$ be the finite field of order $q$, where $q$ is a prime power. The projective general linear group $\mathrm{PGL}(2, q)$ is defined as the quotient of the general linear group $\mathrm{GL}(2, q)$ by its center $Z(\mathrm{GL}(2, q))$:
\[
\mathrm{PGL}(2, q) = \mathrm{GL}(2, q) / Z(\mathrm{GL}(2, q)).
\]
Here,
\[
\mathrm{GL}(2, q) = \left\{ \begin{pmatrix} a & b \\ c & d \end{pmatrix} \mid a, b, c, d \in \mathbb{F}_{q},\ \det\begin{pmatrix} a & b \\ c & d \end{pmatrix} = ad - bc \neq 0 \right\},
\]
and its center is
\[
Z(\mathrm{GL}(2, q)) = \left\{ \lambda \begin{pmatrix} 1 & 0 \\ 0 & 1 \end{pmatrix} \mid \lambda \in \mathbb{F}_{q}^* \right\}.
\]

The projective line $\mathrm{PG}(1,q)$ can be identified with
$
\mathbb{F}_{q}\cup\{\infty\}.
$
Its points may be represented by homogeneous coordinates, where
\[
\infty \longleftrightarrow
\mathbb{F}_q^*\cdot\begin{pmatrix}
1\\
0
\end{pmatrix},
\qquad
a \longleftrightarrow
\mathbb{F}_q^*\cdot\begin{pmatrix}
a\\
1
\end{pmatrix},
\quad a\in\mathbb{F}_{q}.
\]

The action of $\mathrm{PGL}(2, q)$ on $\mathrm{PG}(1, q)$ is given by fractional linear transformations: for any $\begin{pmatrix} a & b \\ c & d \end{pmatrix} \in \mathrm{PGL}(2, q)$ and $x \in \mathrm{PG}(1, q)$,
\[
x \mapsto \frac{ax + b}{cx + d}
\]
with the following conventions:
\begin{enumerate}
    \item If $x = \infty$ and $c \neq 0$, then $\frac{a\infty + b}{c\infty + d} = \frac{a}{c}$;
    \item If $x = \infty$ and $c = 0$, then $\frac{a\infty + b}{c\infty + d} = \infty$;
    \item If the denominator $cx + d = 0$, then $\frac{ax + b}{cx + d} = \infty$.
\end{enumerate}

It is straightforward to verify the associativity of the group action: for any $g_1, g_2 \in \mathrm{PGL}(2, q)$ and $x \in \mathrm{PG}(1, q)$, $g_1(g_2(x)) = (g_1g_2)(x)$.

The group $\mathrm{PGL}(2, q)$ acts sharply $3$-transitively on $\mathrm{PG}(1, q)$; for completeness, we include a short proof. Take any two ordered triples of distinct points $(a, b, c)$ and $(\bar{a}, \bar{b}, \bar{c})$ in $\mathrm{PG}(1, q)$, and construct elements
\[
g_1 = \begin{pmatrix} a(b - c) & b(c - a) \\ b - c & c - a \end{pmatrix}, \quad g_2 = \begin{pmatrix} \bar{a}(\bar{b} - \bar{c}) & \bar{b}(\bar{c} - \bar{a}) \\ \bar{b} - \bar{c} & \bar{c} - \bar{a} \end{pmatrix} \in \mathrm{PGL}(2, q).
\]
{If one of $a,b,c$ (or $\bar a,\bar b,\bar c$) is $\infty$, the corresponding matrix is understood via a projectively equivalent limiting representative; for example, when $a=\infty$, one may take $g_1=\begin{pmatrix}b-c&-b\\0&-1\end{pmatrix}$, and the other cases are analogous.}
One checks directly that $g_1$ maps $(\infty,0,1)$ to $(a,b,c)$ and that $g_2$ maps $(\infty,0,1)$ to $(\bar{a},\bar{b},\bar{c})$. Then $g = g_2g_1^{-1}$ satisfies $g(a) = \bar{a}$, $g(b) = \bar{b}$, $g(c) = \bar{c}$. Moreover, a fractional linear transformation is uniquely determined by the images of three distinct points. Thus, the action of $\mathrm{PGL}(2, q)$ on $\mathrm{PG}(1, q)$ is sharply $3$-transitive.

Let \(U_{q+1}=\{x\in \mathbb{F}_{q^2}^{*}:x^{q+1}=1\}\), which is the cyclic subgroup of \(\mathbb{F}_{q^2}^*\) consisting of all \((q+1)\)-th roots of unity.
Then $U_{q+1}$ can be seen as the subset of the projective line $\mathrm{PG}(1,q^2)=\mathbb{F}_{q^2}\cup\{\infty\}$ consisting of all \((q+1)\)-th roots of unity. Let
$\mathrm{Stab}_{U_{q+1}}$ denote the setwise stabilizer of $U_{q+1}$ under the action of $\mathrm{PGL}(2,\mathbb{F}_{q^2})$ on $\mathrm{PG}(1,\mathbb{F}_{q^2})$. In \cite{LiuDingMesnagerTangTonchev2}, the authors showed the following result.

\begin{lemma}\label{lem1}\cite[Proposition 6]{LiuDingMesnagerTangTonchev2} The setwise stabilizer $\mathrm{Stab}_{U_{q+1}}$ of $U_{q+1}$ can be given by
\[
\left\{ \begin{pmatrix} a^q & b^q \\ b & a \end{pmatrix} \in \mathrm{PGL}(2, q^2) : a^{q+1} \neq b^{q+1} \right\}.
\]
And the action of $\mathrm{Stab}_{U_{q+1}}$ on $U_{q+1}$ is equivalent to the action of $\mathrm{PGL}(2, q)$
on $\mathrm{PG}(1, q)$. Hence $\mathrm{Stab}_{U_{q+1}}$ is sharply $3$-transitive.
\end{lemma}

We identify the \(q+1\) points of \(\mathrm{PG}(1,q)\) with the coordinate positions of a linear code of length \(q+1\) over \(\mathbb{F}_{q^2}\). Thus, a codeword can be written as \(\ \mathbf{c}=(c_x)_{x\in\mathrm{PG}(1,q)}\), and each element of \(\mathrm{PGL}(2,q)\) induces a permutation of its coordinates. Let \(\mathcal B_i(\mathcal{C})\) denote the set of supports of all codewords of Hamming weight \(i\) in \(\mathcal{C}\). If \(\mathcal B_i(\mathcal{C})\) is invariant under this action, since \(\mathrm{PGL}(2,q)\) acts sharply \(3\)-transitively on \(\mathrm{PG}(1,q)\), the following result gives a sufficient condition for \(\mathcal B_i(\mathcal{C})\) to support a \(3\)-design.

\begin{lemma}\label{lem2} \cite{DingTangTonchevPGL} Let $\mathcal{C}$ be a $[q+1,k]$ linear code over $\mathbb{F}_{q^2}$, and let $\mathcal{B}_i(\mathcal{C})$ be the set of supports of all codewords of Hamming weight $i$ in $\mathcal{C}$. If $\mathcal{B}_i(\mathcal{C})$ is invariant under $\mathrm{PGL}(2, q)$, then the incidence structure $(\mathrm{PG}(1, q), \mathcal{B}_i(\mathcal{C}))$ {is a $3$-design} when $A_i \neq 0$, where $A_i$ denotes the number of codewords with weight $i$ in $\mathcal{C}$.
\end{lemma}

The following theorem, which was developed by Assmus and Mattson in \cite{AssmusMattson},
provides a sufficient condition for {the codewords of a linear code and its dual to support simple $t$-designs}.

\begin{theorem}[Assmus-Mattson Theorem] Let $\mathcal{C}$ be an $[n,k,d]$ code over $\mathbb{F}_{q}$. Let $d^{\perp}$ denote the minimum distance of $\mathcal{C}^{\perp}$. Let $w$ be the largest integer satisfying $w\leq n$ and
$$w-\left\lfloor\frac{w+q-2}{q-1}\right\rfloor < d.$$
Define $w^{\perp}$ analogously using $d^{\perp}$. Let $(A_i)_{i=0}^n$ and $(A_i^{\perp})_{i=0}^n$ denote the weight distribution of $\mathcal{C}$ and $\mathcal{C}^{\perp}$, respectively. Fix a positive integer $t$ with $t<d$, and let $s$ be the number of $i$ with $A_i^{\perp}\neq 0$ for $1 \leq i\leq n-t$. Suppose $s\leq d-t$. Then
\begin{itemize}
\item the codewords of weight $i$ in $\mathcal{C}$ {support} a simple $t$-design provided $A_i\neq 0$ and $d\leq i\leq w$, and
\item the codewords of weight $i$ in $\mathcal{C}^{\perp}$ {support} a simple $t$-design provided $A_i^\perp \neq 0$ and $d^{\perp}\leq i\leq \text{min}\{n-t,w^{\perp}\}$.
\end{itemize}
\end{theorem}

 In order to obtain the parameters of the dual codes of the discussed codes, we need the Pless power moment
identities on linear codes.
\begin{lemma}\label{lem3}Let \( \mathcal{C} \) be an \([n, k]\) linear code over the finite field \( \mathbb{F}_{q} \), and let \( \mathcal{C}^\perp \) denote its dual code. For each integer \( i \), let \( A_i \) and \( A_i^\perp \) denote the number of codewords of Hamming weight \( i \) in \( \mathcal{C} \) and \( \mathcal{C}^\perp \), respectively.  The first four Pless power moments are as follows \cite[p. 259]{HuffmanPless}:
\begin{equation*}
\begin{split}
&\sum_{i=0}^nA_i=q^k;\\
&\sum_{i=0}^niA_i=q^{k-1}(qn-n-A_1^{\perp});\\
&\sum_{i=0}^ni^2A_i=q^{k-2}[(q-1)n(qn-n+1)-(2qn-q-2n+2)A_1^{\perp}+2A_2^{\perp}];\\
&\sum_{i=0}^ni^3A_i=q^{k-3}[(q-1)n(q^2n^2-2qn^2+3qn-q+n^2-3n+2)-(3q^2n^2-3q^2n-\\
&\hskip 54pt 6qn^2+12qn+q^2-6q+3n^2-
9n+6)A_1^{\perp}+6(qn-q-n+2)A_2^{\perp}-6A_3^{\perp}].
\end{split}
\end{equation*}
If $A_1^{\perp}=A_2^{\perp}=A_3^{\perp}=0$, then the fifth Pless power moment is as follows:
\begin{equation*}
\begin{split}
&\sum_{i=0}^ni^4A_i=q^{k-4}[(q-1)n(q^3n^3-3q^2n^3+6q^2n^2-4q^2n+q^2+3qn^3-12qn^2+15qn-\\
&\hskip 54pt 6q-n^3+6n^2-11n+6)
+24A_4^{\perp}].
\end{split}
\end{equation*}
\end{lemma}
{The following lemma on linear codes is well known.}

\begin{lemma}\label{lem4}\cite[Corollary 1.4.14]{DingTangBook}
A linear code has minimum weight $d$ if and only if its parity check matrix
has a set of $d$ linearly dependent columns but no set of $d-1$ linearly dependent columns.
\end{lemma}

\section{3-designs from the linear code \(\mathcal{C}\) and its applications}

In this section, we mainly {determine the weight distribution of $\mathcal{C}$, show the connection between $\mathcal{C}$ and combinatorial $3$-designs, and further study some applications of $\mathcal{C}$}, where $\mathcal{C}$ is given in (\ref{eq:C}). We start with the following lemma.

\begin{lemma}\label{lem8}\cite{WangTangDing}
Let $p$ be an odd prime, $s \ge 1$ and $m \ge 2$ be positive integers. Let $q = p^m$, $x,y$ be two distinct elements in $U_{q+1} \setminus \{1\}$. Let
\begin{equation*}
T = \begin{vmatrix}
1, & 1, & 1 \\
1, & x, & y \\
1, & x^{p^s}, & y^{p^s}
\end{vmatrix}.
\end{equation*}
When $v_2(s) \le v_2(m)$, then the determinant $T \neq 0$. When $v_2(s) > v_2(m)$, then $T \neq 0$ if $x,y \notin U_{p^{\gcd(s,m)}+1}$.
\end{lemma}

\begin{lemma}\label{lem7}
Let $q=p^m$ be an odd prime power and $\mathcal{C}$ be the linear code over $\mathbb{F}_{q^2}$ given in (\ref{eq:C}), where $m\geq 2$. Then
$
d(\mathcal{C}^{\perp})\geq 4.
$
\end{lemma}

\begin{proof}
Let $d$ denote the minimum Hamming distance of $\mathcal{C}^{\perp}$. Clearly, $d\geq 2$. We prove that $d\neq 2$ and $d\neq 3$.

Suppose first that $d=2$. Then there exist $a_1\in\mathbb{F}_{q^2}^{*}$ and $1\leq i\leq q$ such that
\begin{equation}\label{eq:oct1027}
\begin{cases}
1+a_1\beta^{\frac{(p^s-1)i}{2}}=0,\\
1+a_1\beta^{\frac{(p^s+1)i}{2}}=0,\\
1+a_1\beta^{-\frac{(p^s-1)i}{2}}=0,\\
1+a_1\beta^{-\frac{(p^s+1)i}{2}}=0.
\end{cases}
\end{equation}
The first two equations in (\ref{eq:oct1027}) give
$a_1\beta^{\frac{(p^s-1)i}{2}}=a_1\beta^{\frac{(p^s+1)i}{2}}$. Since $a_1\neq0$, we obtain $\beta^i=1$. However, $\beta$ has order $q+1$ and $1\leq i\leq q$, a contradiction. Hence $d\neq2$.

Suppose now that $d=3$. Then there exist $a_1,a_2\in\mathbb{F}_{q^2}^{*}$ and $1\leq i_1\neq i_2\leq q$ such that
\begin{equation}\label{eq:oct1019}
\begin{cases}
1+a_1\beta^{\frac{(p^s-1)i_1}{2}}
+a_2\beta^{\frac{(p^s-1)i_2}{2}}=0,\\
1+a_1\beta^{\frac{(p^s+1)i_1}{2}}
+a_2\beta^{\frac{(p^s+1)i_2}{2}}=0,\\
1+a_1\beta^{-\frac{(p^s-1)i_1}{2}}
+a_2\beta^{-\frac{(p^s-1)i_2}{2}}=0,\\
1+a_1\beta^{-\frac{(p^s+1)i_1}{2}}
+a_2\beta^{-\frac{(p^s+1)i_2}{2}}=0.
\end{cases}
\end{equation}
Let $x=\beta^{i_1}$ and $y=\beta^{i_2}$. Then $x$ and $y$ are distinct elements of $U_{q+1}\setminus\{1\}$, and (\ref{eq:oct1019}) becomes
\begin{equation}\label{eq:1013}
\begin{cases}
1+a_1x^{\frac{p^s-1}{2}}+a_2y^{\frac{p^s-1}{2}}=0,\\
1+a_1x^{\frac{p^s+1}{2}}+a_2y^{\frac{p^s+1}{2}}=0,\\
1+a_1x^{-\frac{p^s-1}{2}}+a_2y^{-\frac{p^s-1}{2}}=0,\\
1+a_1x^{-\frac{p^s+1}{2}}+a_2y^{-\frac{p^s+1}{2}}=0.
\end{cases}
\end{equation}

Taking the first, third, and fourth equations in (\ref{eq:1013}), we obtain
\[
\begin{vmatrix}
1 & x^{-\frac{p^s+1}{2}} & y^{-\frac{p^s+1}{2}}\\
1 & x^{-\frac{p^s-1}{2}} & y^{-\frac{p^s-1}{2}}\\
1 & x^{\frac{p^s-1}{2}} & y^{\frac{p^s-1}{2}}
\end{vmatrix}=0.
\]
Multiplying the second and third columns by $x^{\frac{p^s+1}{2}}$ and $y^{\frac{p^s+1}{2}}$, respectively, gives
\begin{equation}\label{eq:1031}
\begin{vmatrix}
1&1&1\\
1&x&y\\
1&x^{p^s}&y^{p^s}
\end{vmatrix}=0.
\end{equation}

Let $\ell=\gcd(m,s)$. By Lemma~\ref{lem8}, (\ref{eq:1031}) implies $v_2(s)>v_2(m)$, and at least one of $x$ and $y$ belongs to $U_{p^\ell+1}$. Without loss of generality, assume $x\in U_{p^\ell+1}$. Since $v_2(s)>v_2(m)$, the integer $s/\ell$ is even, and hence $p^\ell+1$ divides $p^s-1$. Thus $x^{p^s}=x$. Substituting this into (\ref{eq:1031}), we obtain
\[
(x-1)(y^{p^s}-y)=0.
\]
Since $x\neq1$, it follows that $y^{p^s}=y$, and hence $y^{p^s-1}=1$. Moreover,
$\gcd(p^s-1,q+1)=p^\ell+1$ when $v_2(s)>v_2(m)$. Since $y\in U_{q+1}$, we obtain $y\in U_{p^\ell+1}$. Therefore,
\[
x,y\in U_{p^\ell+1}.
\]

Since $s/\ell$ is even, $(p^s-1)/(p^\ell+1)$ is even. Hence, for $z\in U_{p^\ell+1}$,
$z^{\frac{p^s-1}{2}}=1$ and $z^{\frac{p^s+1}{2}}=z$. Therefore, (\ref{eq:1013}) reduces to
\[
\begin{cases}
1+a_1+a_2=0,\\
1+a_1x+a_2y=0,\\
1+a_1x^{-1}+a_2y^{-1}=0.
\end{cases}
\]
Solving the first two equations gives
$a_1=\frac{y-1}{x-y}$ and $a_2=\frac{x-1}{y-x}$. Substituting these into the third equation yields
\[
1+\frac{y-1}{x(x-y)}+\frac{x-1}{y(y-x)}=0.
\]
Since $x\neq y$, simplifying gives $xy=x+y-1$, or equivalently $(x-1)(y-1)=0$. This contradicts $x,y\neq1$. Hence $d\neq3$.

Hence, we have $d(\mathcal{C}^{\perp})\geq4$. This completes the proof.
\end{proof}

In order to obtain the parameters of $\mathcal{C}$ and its dual, we need the following lemma, which is given in  \cite{WangTangDing}; see also \cite{Helleseth2004} for related results on equations over finite fields.
\begin{lemma}\cite{WangTangDing}\label{conj-21march338}
Let $(a,b,c,d)
\in \mathbb{F}_{q^2}^4\setminus \{(0,0,0,0)\}$. When $y$ runs over $U_{q+1}$, {the possible numbers of solutions to the equation}
\begin{equation*}
ay+by^{p^s}+cy^{p^s+1}+d=0
\end{equation*}
{are} $0$, $1$, $2$ or $p^{\gcd(s,m)}+1$.
\end{lemma}

\begin{theorem}\label{lem10}
Let $q=p^m$, where $p$ is an odd prime and $1\leq s<m$, and let $\ell=\gcd(m,s)$. Let $\mathcal{C}$ be the linear code over $\mathbb{F}_{q^2}$ given in (\ref{eq:C}). Then $\mathcal{C}$ has parameters $[q+1,4,q-p^\ell]_{q^2}$ and weight enumerator

\begin{equation*}
\begin{split}
&1+\frac{(q^2-1)^2q}{p^{3\ell} - p^\ell}z^{q-p^\ell}+\frac{(2p^\ell q^4+2q^4+p^\ell q^2+q^2+p^\ell-q)q(q-1)^2(q+1)}{2(1+p^\ell)}z^{q+1}\\
&+\frac{q((q^2+1)p^\ell- q^2-q)(q -1)(q+1)^2}{2(p^\ell-1)}z^{q-1}+((q^5-2q^3+q)p^{-\ell}+q^7-q^5+q^2-1)z^q.
\end{split}
\end{equation*}
Moreover, $\mathcal{C}^{\perp}$ has parameters ${[q+1,q-3,4]_{q^2}}$.
\end{theorem}

\begin{proof}
Let $\boldsymbol{\alpha}_i$ be the $i$-th row of the matrix $A$ given in (\ref{eq:C070201}). We first prove that $\dim_{\mathbb{F}_{q^2}}(\mathcal{C})=4$. Suppose that $a\boldsymbol{\alpha}_1+b\boldsymbol{\alpha}_2+c\boldsymbol{\alpha}_3+d\boldsymbol{\alpha}_4=\boldsymbol{0}$, where $a,b,c,d\in\mathbb{F}_{q^2}$. Then, for every $y\in U_{q+1}$,
\[
ay^{\frac{p^s-1}{2}}+by^{\frac{p^s+1}{2}}+cy^{-\frac{p^s-1}{2}}+dy^{-\frac{p^s+1}{2}}=0.
\]
Multiplying by $y^{\frac{p^s+1}{2}}$, we obtain $cy+ay^{p^s}+by^{p^s+1}+d=0$. If $(a,b,c,d)\neq(0,0,0,0)$, then by Lemma \ref{conj-21march338}, this equation has at most $p^\ell+1$ solutions in $U_{q+1}$. Since $s<m$, we have $\ell<m$, and hence $p^\ell+1<q+1$. Therefore, the above equation cannot vanish at all $q+1$ elements of $U_{q+1}$, a contradiction. Hence $a=b=c=d=0$, and thus $\dim_{\mathbb{F}_{q^2}}(\mathcal{C})=4$.

For $\mathbf{c}=(a,b,c,d)A\in\mathcal{C}$,
\[
\operatorname{wt}(\mathbf{c})=q+1-\left|\left\{y\in U_{q+1}: ay^{\frac{p^s-1}{2}}+by^{\frac{p^s+1}{2}}+cy^{-\frac{p^s-1}{2}}+dy^{-\frac{p^s+1}{2}}=0\right\}\right|.
\]
As above, the zero equation is equivalent to $cy+ay^{p^s}+by^{p^s+1}+d=0$. Hence, by Lemma \ref{conj-21march338}, the possible numbers of zeros are $0,1,2$, and $p^\ell+1$. {Therefore, every nonzero weight of $\mathcal{C}$ belongs to $\{q+1,q,q-1,q-p^\ell\}$. The equality $d(\mathcal{C})=q-p^\ell$ will follow once we show that $A_{q-p^\ell}>0$.}

Let $A_{q+1},A_q,A_{q-1}$, and $A_{q-p^\ell}$ denote the numbers of codewords of the corresponding weights. By Lemma \ref{lem7}, $d(\mathcal{C}^{\perp})\geq4$, and hence $A_1^\perp=A_2^\perp=A_3^\perp=0$. Applying the first four Pless power moments in Lemma \ref{lem3} to the $[q+1,4]_{q^2}$ code $\mathcal{C}$, we obtain
\[
\begin{cases}
A_{q+1}+A_q+A_{q-1}+A_{q-p^\ell}=q^8-1,\\
(q+1)A_{q+1}+qA_q+(q-1)A_{q-1}+(q-p^\ell)A_{q-p^\ell}=q^6(q^2-1)(q+1),\\
(q+1)^2A_{q+1}+q^2A_q+(q-1)^2A_{q-1}+(q-p^\ell)^2A_{q-p^\ell}=q^4(q^2-1)(q+1)(q^3+q^2-q),\\
(q+1)^3A_{q+1}+q^3A_q+(q-1)^3A_{q-1}+(q-p^\ell)^3A_{q-p^\ell}\\
\qquad\qquad=q^2(q^2-1)(q+1)\bigl(q^4(q+1)^2-2q^2(q+1)^2+3q^2(q+1)-q\bigr).
\end{cases}
\]
Solving this system gives
\[
\begin{cases}
\displaystyle A_{q-p^\ell}=\frac{(q^2-1)^2q}{p^{3\ell}-p^\ell},\\[2mm]
\displaystyle A_{q-1}=\frac{q((q^2+1)p^\ell-q^2-q)(q-1)(q+1)^2}{2(p^\ell-1)},\\[2mm]
\displaystyle A_q=(q^5-2q^3+q)p^{-\ell}+q^7-q^5+q^2-1,\\[2mm]
\displaystyle A_{q+1}=\frac{(2p^\ell q^4+2q^4+p^\ell q^2+q^2+p^\ell-q)q(q-1)^2(q+1)}{2(1+p^\ell)}.
\end{cases}
\]
{Since $A_{q-p^\ell}=\dfrac{(q^2-1)^2q}{p^{3\ell}-p^\ell}>0$, the weight $q-p^\ell$ actually occurs, and hence $d(\mathcal C)=q-p^\ell$.}

It remains to determine the minimum distance of $\mathcal{C}^{\perp}$. Since $A_1^\perp=A_2^\perp=A_3^\perp=0$, the fifth Pless power moment in Lemma \ref{lem3}, together with the above weight distribution, yields
\[
A_4^\perp=\frac{q(q^2-1)^2(p^\ell-2)}{24}.
\]
Since $p$ is odd and $\ell\geq1$, we have $p^\ell\geq3$, so $A_4^\perp>0$. Thus $d(\mathcal{C}^{\perp})\leq4$. Combining this with Lemma \ref{lem7}, which gives $d(\mathcal{C}^{\perp})\geq4$, we obtain $d(\mathcal{C}^{\perp})=4$. Finally, $\dim(\mathcal{C}^{\perp})=(q+1)-4=q-3$, and hence $\mathcal{C}^{\perp}$ has parameters $[q+1,q-3,4]_{q^2}$. This completes the proof.
\end{proof}

\begin{proposition}\label{lem11}
Let the symbols be given as above. Let $\mathcal{B}_i(\mathcal{C})$ be the family of supports of all codewords of Hamming weight $i$ in $\mathcal{C}$. If $|\mathcal{B}_i(\mathcal{C})|\neq0$, then $(U_{q+1},\mathcal{B}_i(\mathcal{C}))$ supports a $3$-design.
\end{proposition}

\begin{proof}
{By Lemma \ref{lem1}, $\operatorname{Stab}_{U_{q+1}}$ acts sharply $3$-transitively on $U_{q+1}$. In view of Lemma \ref{lem2}, it is enough to prove that the family of supports $\mathcal B_i(\mathcal C)$ is invariant under this action. Thus, for each $\pi\in\operatorname{Stab}_{U_{q+1}}$ and each $\mathbf c\in\mathcal C$, it suffices to find a codeword $\overline{\mathbf c}\in\mathcal C$ such that
\[
\operatorname{supp}(\pi(\mathbf c))=\operatorname{supp}(\overline{\mathbf c}).
\]
}
From the definition of $\operatorname{Stab}_{U_{q+1}}$, $\pi$ can be expressed as
 \(\pi=\begin{pmatrix}a^q&b^q\\ b&a\end{pmatrix}\in \operatorname{Stab}_{U_{q+1}}\). For any \(x\in U_{q+1}\), represent the corresponding projective point by
$
\binom{x}{1}.
$
Then
\[
\pi\binom{x}{1}
=
\binom{a^qx+b^q}{bx+a}
=
(bx+a)
\binom{\dfrac{a^qx+b^q}{bx+a}}{1}.
\]
Therefore, the action induced by \(\pi\) on \(U_{q+1}\) is
\begin{equation}\label{eq:pib}
\pi\, :\, x\longmapsto \frac{a^qx+b^q}{bx+a}.
\end{equation}
By definition, for any codeword $\mathbf{c}\in \mathcal{C}$, we have
\[
\mathbf{c}(c_1,c_2,c_3,c_4)=\left(
c_1x^{\frac{p^s+1}{2}}
+c_2x^{\frac{p^s-1}{2}}
+c_3x^{-\frac{p^s+1}{2}}
+c_4x^{-\frac{p^s-1}{2}}
\right)_{x \in U_{q+1}},
\]
where $c_1,c_2,c_3,c_4 \in \mathbb{F}_{q^2}$. {Let
\[
f(x)=c_1x^{\frac{p^s+1}{2}}+c_2x^{\frac{p^s-1}{2}}+c_3x^{-\frac{p^s+1}{2}}+c_4x^{-\frac{p^s-1}{2}}.
\]
Then $\mathbf{c}(c_1,c_2,c_3,c_4)=(f(x))_{x\in U_{q+1}}$. We verify the required support invariance by considering three cases.}

\textbf{Case 1}: $a\neq0$ and $b=0$.
In this case, (\ref{eq:pib}) becomes $\pi:x\mapsto a^{q-1}x$, where $a^{q-1}\in U_{q+1}$. Therefore,
\[
f(a^{q-1}x)
=c_1(a^{q-1})^{\frac{p^s+1}{2}}x^{\frac{p^s+1}{2}}
+c_2(a^{q-1})^{\frac{p^s-1}{2}}x^{\frac{p^s-1}{2}}
+c_3(a^{q-1})^{-\frac{p^s+1}{2}}x^{-\frac{p^s+1}{2}}
+c_4(a^{q-1})^{-\frac{p^s-1}{2}}x^{-\frac{p^s-1}{2}}.
\]
Hence,
\[
\pi(\mathbf{c}(c_1,c_2,c_3,c_4))
=
\mathbf{c}\left(
c_1(a^{q-1})^{\frac{p^s+1}{2}},
c_2(a^{q-1})^{\frac{p^s-1}{2}},
c_3(a^{q-1})^{-\frac{p^s+1}{2}},
c_4(a^{q-1})^{-\frac{p^s-1}{2}}
\right)\in\mathcal{C}.
\]

\textbf{Case 2}: $a=0$ and $b\neq0$.
In this case, (\ref{eq:pib}) becomes $\pi:x\mapsto b^{q-1}x^{-1}$. A direct calculation gives
\[
f(b^{q-1}x^{-1})
=c_3(b^{1-q})^{\frac{p^s+1}{2}}x^{\frac{p^s+1}{2}}
+c_4(b^{1-q})^{\frac{p^s-1}{2}}x^{\frac{p^s-1}{2}}
+c_1(b^{1-q})^{-\frac{p^s+1}{2}}x^{-\frac{p^s+1}{2}}
+c_2(b^{1-q})^{-\frac{p^s-1}{2}}x^{-\frac{p^s-1}{2}}.
\]
Thus,
\[
\pi(\mathbf{c}(c_1,c_2,c_3,c_4))
=
\mathbf{c}\left(
c_3(b^{1-q})^{\frac{p^s+1}{2}},
c_4(b^{1-q})^{\frac{p^s-1}{2}},
c_1(b^{1-q})^{-\frac{p^s+1}{2}},
c_2(b^{1-q})^{-\frac{p^s-1}{2}}
\right)\in\mathcal{C}.
\]

\textbf{Case 3}: $a\neq0$ and $b\neq0$.
Let $A=a^qx+b^q$ and $B=bx+a$. Then
\[
f\left(\frac{A}{B}\right)
=
\frac{c_1A^{p^s+1}+c_2A^{p^s}B+c_3B^{p^s+1}+c_4AB^{p^s}}
{A^{\frac{p^s+1}{2}}B^{\frac{p^s+1}{2}}}.
\]
Since
$A^{p^s}=a^{qp^s}x^{p^s}+b^{qp^s}$ and
$B^{p^s}=b^{p^s}x^{p^s}+a^{p^s}$, we have
\[
\begin{aligned}
A^{p^s+1}&=a^{q(p^s+1)}x^{p^s+1}+a^{qp^s}b^qx^{p^s}+a^qb^{qp^s}x+b^{q(p^s+1)},\\
A^{p^s}B&=a^{qp^s}bx^{p^s+1}+a^{qp^s+1}x^{p^s}+b^{qp^s+1}x+ab^{qp^s},\\
B^{p^s+1}&=b^{p^s+1}x^{p^s+1}+ab^{p^s}x^{p^s}+a^{p^s}bx+a^{p^s+1},\\
AB^{p^s}&=a^qb^{p^s}x^{p^s+1}+b^{q+p^s}x^{p^s}+a^{q+p^s}x+a^{p^s}b^q.
\end{aligned}
\]
Hence,
\[
c_1A^{p^s+1}+c_2A^{p^s}B+c_3B^{p^s+1}+c_4AB^{p^s}
=\bar c_1x^{p^s+1}+\bar c_2x^{p^s}+\bar c_4x+\bar c_3,
\]
where
\[
\begin{aligned}
\bar c_1&=c_1a^{q(p^s+1)}+c_2a^{qp^s}b+c_3b^{p^s+1}+c_4a^qb^{p^s},\\
\bar c_2&=c_1a^{qp^s}b^q+c_2a^{qp^s+1}+c_3ab^{p^s}+c_4b^{q+p^s},\\
\bar c_3&=c_1b^{q(p^s+1)}+c_2ab^{qp^s}+c_3a^{p^s+1}+c_4a^{p^s}b^q,\\
\bar c_4&=c_1a^qb^{qp^s}+c_2b^{qp^s+1}+c_3a^{p^s}b+c_4a^{q+p^s}.
\end{aligned}
\]
Since
\[
\bar c_1x^{p^s+1}+\bar c_2x^{p^s}+\bar c_4x+\bar c_3
=
x^{\frac{p^s+1}{2}}
\left(
\bar c_1x^{\frac{p^s+1}{2}}
+\bar c_2x^{\frac{p^s-1}{2}}
+\bar c_3x^{-\frac{p^s+1}{2}}
+\bar c_4x^{-\frac{p^s-1}{2}}
\right),
\]
we obtain
\[
f\left(\frac{a^qx+b^q}{bx+a}\right)
=
\rho(x)
\left(
\bar c_1x^{\frac{p^s+1}{2}}
+\bar c_2x^{\frac{p^s-1}{2}}
+\bar c_3x^{-\frac{p^s+1}{2}}
+\bar c_4x^{-\frac{p^s-1}{2}}
\right),
\]
where
\[
\rho(x)=
\frac{x^{\frac{p^s+1}{2}}}
{(a^qx+b^q)^{\frac{p^s+1}{2}}(bx+a)^{\frac{p^s+1}{2}}}.
\]
Since $\pi\in\operatorname{Stab}_{U_{q+1}}$, we have $a^{q+1}\neq b^{q+1}$, which implies $bx+a\neq0$ and $a^qx+b^q=x(bx+a)^q\neq0$ for all $x\in U_{q+1}$, and hence $\rho(x)\neq0$. Therefore, letting $\overline{\mathbf c}=\mathbf c(\bar c_1,\bar c_2,\bar c_3,\bar c_4)\in\mathcal C$, the above equality gives $\operatorname{supp}(\pi(\mathbf c))=\operatorname{supp}(\overline{\mathbf c})$.

{In summary, for every $\pi\in\operatorname{Stab}_{U_{q+1}}$ and every $\mathbf c\in\mathcal C$, there exists $\overline{\mathbf c}\in\mathcal C$ with $\operatorname{supp}(\pi(\mathbf c))=\operatorname{supp}(\overline{\mathbf c})$. Hence $\mathcal B_i(\mathcal C)$ is invariant under $\operatorname{Stab}_{U_{q+1}}$. By Lemmas \ref{lem1} and \ref{lem2}, $(U_{q+1},\mathcal B_i(\mathcal C))$ is a $3$-design whenever $\mathcal B_i(\mathcal C)\neq\varnothing$. This completes the proof.}
\end{proof}

\begin{remark}\label{rem:additiveC2}
In \cite{WangTangDing}, the authors proved that the minimum-weight codewords of $\mathcal{C}_1$ support a $3$-design.
 Using the same method as in Proposition \ref{lem11}, we can also show that the codewords of every nonzero weight in $\mathcal{C}_1$ support a $3$-design.
\end{remark}

\begin{theorem}\label{theorem1}
Let the symbols be given as in Lemma \ref{lem10}. Then the minimum-weight codewords of $\mathcal{C}$ support a $3$-$(q+1,q-p^\ell,\lambda)$ design, where
\[
\lambda=
\frac{(q-p^\ell)(q-p^\ell-1)(q-p^\ell-2)}
{p^{3\ell}-p^\ell}.
\]
Moreover, the minimum-weight codewords of $\mathcal{C}^{\perp}$ support a $3$-$(q+1,4,p^\ell-2)$ design.
\end{theorem}

\begin{proof}
By Proposition \ref{lem11} and Theorem \ref{lem10}, the minimum weight of $\mathcal{C}$ is $q-p^\ell$, and the codewords of weight $q-p^\ell$ support a $3$-$(q+1,q-p^\ell,\lambda)$ design. Hence, the number of blocks is $$b=\frac{A_{q-p^\ell}}{q^2-1}=\frac{(q^2-1)q}{p^{3\ell}-p^\ell}.$$ By (\ref{eq:design}), we obtain
$
\lambda\binom{q+1}{3}
=
b\binom{q-p^\ell}{3},
$
and therefore,
\[
\lambda
=
\frac{b\binom{q-p^\ell}{3}}{\binom{q+1}{3}}
=
\frac{(q-p^\ell)(q-p^\ell-1)(q-p^\ell-2)}
{p^{3\ell}-p^\ell}.
\]

We next consider $\mathcal{C}^{\perp}$. Since $\mathcal{C}^{\perp}=[q+1,q-3,4]$, by the Assmus--Mattson theorem, the codewords of weight $4$ in $\mathcal{C}^{\perp}$ support a $3$-$(q+1,4,\lambda^\perp)$ design. The number of blocks is $b^\perp=\frac{A_4^\perp}{q^2-1}=\frac{q(q^2-1)(p^\ell-2)}{24}$. By (\ref{eq:design}),
$
\lambda^\perp\binom{q+1}{3}
=
b^\perp\binom{4}{3}.
$
Thus,
\[
\lambda^\perp
=
\frac{q(q^2-1)(p^\ell-2)}{24}
\frac{4}{\binom{q+1}{3}}
=
p^\ell-2.
\]
We obtain, the minimum-weight codewords of $\mathcal{C}^{\perp}$ support a $3$-$(q+1,4,p^\ell-2)$ design.
\end{proof}

In the following, we give some examples for Theorem \ref{lem10}  and Theorem \ref{theorem1}. These results are consistent with those obtained by Magma.

\begin{example}
Let $p=3$, $s=1$ and $m=2$. Then $q=9$, $\ell=\gcd(m,s)=1$, and
$\mathcal{C}$ has parameters $[10,4,6]$ and weight enumerator
\[
1+2400x^6+280800x^8+4743200x^9+38020320x^{10}.
\]
The codewords of minimum weight $6$ in $\mathcal{C}$ support a
$3$-$(10,6,5)$ design, and the codewords of minimum weight $4$ in $\mathcal{C}^{\perp}$ support a
$3$-$(10,4,1)$ design.
\end{example}

\begin{example}
Let $p=5$, $s=1$ and $m=2$. Then $q=25$, $\ell=\gcd(m,s)=1$, and
$\mathcal{C}$ has parameters $[26,4,20]$ and weight enumerator
\[
1+81120x^{20}+125736000x^{24}
+6095697504x^{25}+146366376000x^{26}.
\]
The codewords of minimum weight $20$ in $\mathcal{C}$ support a
$3$-$(26,20,57)$ design, and the codewords of minimum weight $4$ in $\mathcal{C}^{\perp}$ support a
$3$-$(26,4,3)$ design.
\end{example}

\begin{example}
Let $p=3$, $s=2$ and $m=4$. Then $q=81$, $\ell=\gcd(m,s)=2$, and
$\mathcal{C}$ has parameters $[82,4,72]$ and weight enumerator
\[
1+4841280x^{72}+142740299520x^{80}
+22873692979520x^{81}+1830003750731520x^{82}.
\]
The codewords of minimum weight $72$ in $\mathcal{C}$ support a
$3$-$(82,72,497)$ design, and the codewords of minimum weight $4$ in $\mathcal{C}^{\perp}$ support a
$3$-$(82,4,7)$ design.
\end{example}

We now investigate the applications of the linear code $\mathcal{C}$ constructed above.
Let $[[n,k,d;c]]_q$ denote a $q$-ary EAQECC which encodes $k$ logical $q$-ary quantum symbols into $n$ physical symbols with minimum distance $d$ by means of $c$ pairs of maximally entangled states. For a matrix $H$ over $\mathbb{F}_{q^2}$, its Hermitian conjugate transpose is denoted by $H^\dagger=(H^q)^T$.

\begin{lemma}\label{lem:EAQECC}
{\cite{WildeBrun}} Let $D$ be an {$[n,k,d]_{q^2}$ linear code} with parity-check matrix $H$. Then there exists a $q$-ary EAQECC with parameters $[[n,2k-n+c,d;c]]_q$, where $c=\operatorname{rank}(HH^\dagger)$.
\end{lemma}
If $c=n-k$, then the corresponding EAQECC is called a maximal-entanglement EAQECC. The rate and net rate of an EAQECC $[[n,k,d;c]]_q$ are defined as $k/n$ and $(k-c)/n$.

\begin{theorem}\label{thm:EAQECC}
There exists a $q$-ary EAQECC with parameters $[[q+1,q-3,4;4]]_q$. Moreover, this EAQECC is a maximal-entanglement EAQECC.
\end{theorem}

\begin{proof}
Since $\dim(\mathcal{C})=4$, $A$ is also a parity-check matrix of $\mathcal{C}^{\perp}$. Recall that $\boldsymbol{\alpha}_i$ denotes the $i$-th row of $A$, where $1\leq i\leq4$. Since $\beta^{q+1}=1$, we have $\beta^q=\beta^{-1}$.

For $1\leq i,j\leq4$, consider the Hermitian inner product $\langle\boldsymbol{\alpha}_i,\boldsymbol{\alpha}_j\rangle_H$. If $i=j$, then $\langle\boldsymbol{\alpha}_i,\boldsymbol{\alpha}_i\rangle_H=q+1=1$ in $\mathbb{F}_{q^2}$. If $i\neq j$, by the explicit form of the four rows of $A$, the exponent difference occurring in $\langle\boldsymbol{\alpha}_i,\boldsymbol{\alpha}_j\rangle_H$ is one of $\pm1$, $\pm(p^s-1)$, $\pm p^s$, and $\pm(p^s+1)$. Since $s<m$, all these integers are nonzero modulo $q+1$. Hence, by the geometric-series formula, $\langle\boldsymbol{\alpha}_i,\boldsymbol{\alpha}_j\rangle_H=0$. Therefore, $AA^\dagger=I_4$ and $\operatorname{rank}(AA^\dagger)=4$.

Applying Lemma \ref{lem:EAQECC} to {$\mathcal{C}^{\perp}=[q+1,q-3,4]_{q^2}$} gives a $q$-ary EAQECC with parameters
\[
[[q+1,2(q-3)-(q+1)+4,4;4]]_q=[[q+1,q-3,4;4]]_q.
\]
Since $4=(q+1)-(q-3)$, the resulting EAQECC is a maximal-entanglement EAQECC.
\end{proof}

\begin{remark}
The EAQECC in Theorem \ref{thm:EAQECC} has rate $(q-3)/(q+1)$ and net rate $(q-7)/(q+1)$. In particular, its net rate approaches $1$ as $q$ tends to infinity. Hence, the above construction gives a family of maximal-entanglement EAQECCs with minimum distance $4$ and asymptotically high net rate.
\end{remark}

In the following, we give some examples for Theorem \ref{thm:EAQECC} . These results are consistent with those obtained by Magma.

\begin{example}
Let $p=3$, $m=2$, and $s=1$. Then $q=9$. By Theorem \ref{thm:EAQECC}, there exists a maximal-entanglement EAQECC with parameters
$
{[[10,6,4;4]]_9}.
$

\end{example}

\begin{example}
Let $p=5$, $m=2$, and $s=1$. Then $q=25$. By Theorem \ref{thm:EAQECC},
there exists a maximal-entanglement EAQECC with parameters
$
{[[26,22,4;4]]_{25}}.
$
\end{example}

Locally recoverable codes (LRCs) were introduced to efficiently repair erased symbols in distributed storage systems. The $i$-th coordinate of a linear code is said to have locality $r$ if it can be recovered from at most $r$ other coordinates. A linear code has all-symbol locality $r$ if every coordinate has locality at most $r$ {\cite{GopalanHuangSimitciYekhanin}}.

\begin{lemma}\label{lem:SingletonLRC}
{\cite{GopalanHuangSimitciYekhanin}} Let $D$ be an $[n,k,d]$ linear code with all-symbol locality $r$. Then $d\leq n-k-\left\lceil\frac{k}{r}\right\rceil+2$.
\end{lemma}

An LRC attaining the bound in Lemma \ref{lem:SingletonLRC} is called distance-optimal. In addition to the above Singleton-type bound, Cadambe and Mazumdar {\cite{CadambeMazumdar}} established an upper bound on the dimension of an LRC.

\begin{lemma}\label{lem:CMBound}
{\cite{CadambeMazumdar}} Let $D$ be an $[n,k,d]_{Q}$ linear code with all-symbol locality $r$. Then
\[
k\leq\min_{t\in\mathbb{Z}_{+}}\left\{tr+k_{\mathrm{opt}}^{(Q)}(n-t(r+1),d)\right\},
\]
where $k_{\mathrm{opt}}^{(Q)}(N,d)$ denotes the largest possible dimension of a $Q$-ary linear code of length $N$ and minimum distance $d$.
\end{lemma}

An LRC whose dimension attains the bound in Lemma \ref{lem:CMBound} is called dimension-optimal.

\begin{lemma}\label{lem:locality}\cite{Tan2023}
Let $D$ be a linear code with dual minimum distance $d^\perp$. If every coordinate is contained in the support of some minimum-weight codeword of $D^\perp$, then the minimum all-symbol locality of $D$ is $d^\perp-1$.
\end{lemma}

We now show that $\mathcal{C}$ gives a family of optimal LRCs for a particular choice of the parameters.

\begin{theorem}\label{thm:LRC}
The code $\mathcal{C}$ has all-symbol locality $r=3$. {Moreover, $\mathcal{C}$ is distance-optimal if and only if $p=3$ and $\gcd(m,s)=1$. In this case, $\mathcal{C}$ is also dimension-optimal.}
\end{theorem}

\begin{proof}
By Theorem \ref{theorem1}, the supports of the minimum-weight codewords of $\mathcal{C}^{\perp}$ form a $3$-$(q+1,4,p^\ell-2)$ design. Hence every coordinate is contained in the support of some weight-$4$ codeword of $\mathcal{C}^{\perp}$. Since $d(\mathcal{C}^{\perp})=4$, Lemma \ref{lem:locality} gives $r=3$.

By Lemma \ref{lem:SingletonLRC}, $d\leq n-k-\left\lceil\frac{k}{r}\right\rceil+2=(q+1)-4-\left\lceil\frac{4}{3}\right\rceil+2=q-3$. Since $d(\mathcal{C})=q-p^\ell$, $\mathcal{C}$ attains the Singleton-type bound if and only if $p^\ell=3$. Since $p$ is an odd prime and $\ell=\gcd(m,s)\geq1$, this is equivalent to $p=3$ and $\ell=1$, i.e., $p=3$ and $\gcd(m,s)=1$.

When $p^\ell=3$, we have $d=q-3$. Taking $t=1$ in Lemma \ref{lem:CMBound}, we obtain
\[
k\leq3+k_{\mathrm{opt}}^{(q^2)}(q-3,q-3).
\]
By the classical Singleton bound, $k_{\mathrm{opt}}^{(q^2)}(q-3,q-3)\leq1$, and hence $k\leq4$. Since $\dim(\mathcal{C})=4$, equality holds. Therefore, $\mathcal{C}$ is also dimension-optimal.
\end{proof}

\begin{remark}
When $p=3$ and $\gcd(m,s)=1$, the code $\mathcal{C}$ is both distance-optimal and dimension-optimal. Hence, in this case, $\mathcal{C}$ simultaneously attains the Singleton-type bound and the Cadambe--Mazumdar bound for locally repairable codes.
\end{remark}

In the following, we give some examples for Theorem \ref{thm:LRC}. These results are consistent with those obtained by Magma.

\begin{example}
Let $p=3$, $m=2$, and $s=1$. Then $q=9$ and $\ell=1$. Hence $\mathcal{C}$ is an optimal
$
[10,4,6]
$
LRC with locality $3$.
\end{example}
\begin{example}
Let $p=3$, $m=3$, and $s=1$. Then $q=27$ and
$\ell=1$. Hence $\mathcal{C}$ is an optimal
$
[28,4,24]
$
LRC with locality $3$.
\end{example}

\section{3-design from the nonlinear code \(\mathcal{C}_2\)}

In this section,  we establish the relationship between the supports of codewords in $\mathcal{C}$ and $\mathcal{C}_2$, and determine the parameters of $\mathcal{C}_2$, where $\mathcal{C}_2$ is given in (\ref{eq:C2}). Moreover, we show that the codewords of each nonzero Hamming weight in $\mathcal{C}_2$ support the corresponding $3$-designs.

\begin{theorem}\label{thm:C2}
Let \(q=p^m\), \(1\leq s<m\) and
$\ell=\gcd(m,s),$ where \(p\) is an odd prime.
Let \(\mathcal{C}_2\) be the code over \(\mathbb{F}_{q^2}\) defined in \((\ref{eq:C2})\). Then \(\mathcal{C}_2\) is a nonlinear code over \(\mathbb{F}_{q^2}\) with length \(q+1\) and cardinality \(q^4\). In addition, the Hamming weight of every nonzero codeword of \(\mathcal{C}_2\) belongs to
$$
\left\{q-p^\ell,\;q-1,\;q,\;q+1\right\}.
$$
\end{theorem}
\begin{proof}
Recall from \((\ref{eq:C2})\) that
$$
\mathcal{C}_2
=
\left\{
(a,b,-a^q,-b^q)A:
a,b\in\mathbb{F}_{q^2}
\right\}.
$$
Let
$
\boldsymbol{c}_2=(1,0,-1,0)A\in\mathcal{C}_2
$
and choose \(\mu\in\mathbb{F}_{q^2}\setminus\mathbb{F}_q\).
Suppose that {\(\mu\boldsymbol{c}_2\in\mathcal{C}_2\)}. Then there exist
\(a,b\in\mathbb{F}_{q^2}\) such that
\begin{equation}\label{eq:0906}
(\mu,0,-\mu,0)A=(a,b,-a^q,-b^q)A.
\end{equation}
From Theorem \ref{lem10}, we {know} that the rows of \(A\) are linearly independent. Then, if (\ref{eq:0906}) holds,  we must have
$$
a=\mu,\qquad b=0,\qquad a^q=\mu.
$$
Thus \(\mu^q=\mu\), which implies \(\mu\in\mathbb{F}_q\), a contradiction.
Hence, \(\mathcal{C}_2\) is not closed under scalar multiplication by
\(\mathbb{F}_{q^2}\). Therefore, \(\mathcal{C}_2\) is nonlinear over \(\mathbb{F}_{q^2}\).

Now consider a nonzero codeword
$$
\boldsymbol{c}_2=(a,b,-a^q,-b^q)A\in\mathcal{C}_2,
$$
where \((a,b)\neq(0,0)\). By the definition of \(A\),
\begin{equation}\label{eq:wts}
\operatorname{wt}(\boldsymbol{c}_2)
=q+1-\Bigl|
\bigl\{y\in U_{q+1}:
ay^{\frac{p^s-1}{2}}
+by^{\frac{p^s+1}{2}}-a^qy^{-\frac{p^s-1}{2}}
-b^qy^{-\frac{p^s+1}{2}}
=0
\bigr\}
\Bigr|.
\end{equation}
Multiplying the equation inside the braces by
\(y^{(p^s+1)/2}\), which is nonzero for \(y\in U_{q+1}\), gives
$$
ay^{p^s}+by^{p^s+1}-a^qy-b^q=0.
$$
From Lemma \ref{conj-21march338}, the number of solutions of this equation in
\(U_{q+1}\) is one of $0$, $1$, $2$ or $p^\ell+1$.
From (\ref{eq:wts}), every nonzero codeword of \(\mathcal{C}_2\) has Hamming
weight belonging to
$$
\left\{
q+1,\;q,\;q-1,\;q-p^\ell
\right\}.
$$

We next show that the cardinality of $\mathcal{C}_2$ is $q^4$. Let $\mathbf{c}_1 \in \mathcal{C}_1$ and $\mathbf{c}_2 \in \mathcal{C}_2$. By the definition of $\mathcal{C}_1$ and $\mathcal{C}_2$, we have that
$\mathbf{c}_1$ and $\mathbf{c}_2$ can be expressed as
$$\mathbf{c}_1 = (c,d,c^q,d^q)A\,\, \text{and}\,\,\mathbf{c}_2 = (a,b,-a^q,-b^q)A,$$
where $a,b,c,d \in \mathbb{F}_{q^2}$.
{Since $p$ is odd, the map
\[
(a,b,c,d)\mapsto(a+c,b+d,c^q-a^q,d^q-b^q)
\]
is bijective; indeed, if its image is $(u,v,w,z)$, then
$a=(u-w^q)/2$, $c=(u+w^q)/2$, $b=(v-z^q)/2$, and $d=(v+z^q)/2$.}
Thus
$$\{(a+c,b+d,c^q-a^q,d^q-b^q)A\,|\, a,b,c,d \in \mathbb{F}_{q^2}\}=\{(a,b,c,d)A\,|\, a,b,c,d \in \mathbb{F}_{q^2}\}.$$
Then, by the definition of $\mathcal{C}$, for any $\mathbf{c} \in \mathcal{C}$, there exist $\mathbf{c}_1$ and $\mathbf{c}_2$ such that
$\mathbf{c}=\mathbf{c}_1+\mathbf{c}_2$,
which implies that
\begin{equation}\label{eq:dfC}
\mathcal{C}=
\left\{
\mathbf{c}_1+\mathbf{c}_2:
\mathbf{c}_1\in \mathcal{C}_1,\,
\mathbf{c}_2\in \mathcal{C}_2
\right\}.
\end{equation}
 Moreover, $\mathcal C_1\cap\mathcal C_2=\{\mathbf0\}$. Indeed, if $(c,d,c^q,d^q)A=(a,b,-a^q,-b^q)A$, then the linear independence of the rows of $A$ gives $c=a$, $d=b$, $c^q=-a^q$ and $d^q=-b^q$. Since $p$ is odd, we obtain $a=b=c=d=0$. {Since both $\mathcal C_1$ and $\mathcal C_2$ are additive subgroups of $\mathcal C$, we have} $\mathcal C=\mathcal C_1\oplus\mathcal C_2$, and hence,
\begin{equation}\label{eq:dfC1}
|\mathcal{C}|= |\mathcal{C}_1||\mathcal{C}_2|.
\end{equation}
From Theorem \ref{lem10}, we know that  $\mathcal{C}$ is a linear code with dimension 4 over \(\mathbb{F}_{q^2}\). In \cite{WangTangDing}, the authors proved that $\mathcal{C}_1$ is a cyclic code with dimension 4 over \(\mathbb{F}_{q}\). Hence,  (\ref{eq:dfC1}) implies that the number of codewords in $\mathcal{C}_2$ over  \(\mathbb{F}_{q^2}\) is \((q^2)^2 = q^4\).
This completes the proof.
\end{proof}

\begin{theorem}\label{lem22} Let \(\mathcal{B}_i\) be the family of support sets of all codewords with weight \(i\) in code \(\mathcal{C}_2\). If \(|\mathcal{B}_i|\neq 0\), then \((U_{q+1},\mathcal{B}_i)\) admits a \(3\)-design. Moreover, the codewords with weight $q-p^\ell$ in $\mathcal{C}_2$ support a $3$-$(q+1,q-p^\ell,\lambda')$ design with
$$\lambda'=\frac{
(q-p^\ell)(q-p^\ell-1)(q-p^\ell-2)
}{
p^{3\ell}-p^\ell
}.$$
\end{theorem}

\begin{proof}
By Theorem \ref{thm:C2}, the possible nonzero weights of
\(\mathcal{C}_2\) are contained in
\[
\left\{
q-p^\ell,\ q-1,\ q,\ q+1
\right\}.
\]
Since \(\mathcal{C}_2\subseteq\mathcal{C}\), we immediately have
\begin{equation}\label{eq:ddd97}
\mathcal{B}_i(\mathcal{C}_2)
\subseteq
\mathcal{B}_i(\mathcal{C})
\end{equation}
for every \(i\).
We first consider
$
i\in\{q-p^\ell,q-1,q\}$ and claim that
$
\mathcal{B}_i(\mathcal{C})
\subseteq
\mathcal{B}_i(\mathcal{C}_2).
$

Let
$
\boldsymbol{c}
=
(c_1,c_2,\ldots,c_{q+1})
\in\mathcal{C}
$
be a codeword of weight \(i\). Write
$
\boldsymbol{c}
=
(a_1,a_2,a_3,a_4)A,$
where
$a_1,a_2,a_3,a_4\in\mathbb{F}_{q^2}.$
By the definition of \(A\), for \(1\leq j\leq q+1\),
\[
c_j
=
a_1(\beta^j)^{\frac{p^s-1}{2}}
+
a_2(\beta^j)^{\frac{p^s+1}{2}}
+
a_3(\beta^j)^{-\frac{p^s-1}{2}}
+
a_4(\beta^j)^{-\frac{p^s+1}{2}}.
\]

Since \(\beta^{q+1}=1\), we have
$
\beta^q=\beta^{-1}.
$
Then for every
\(\mu\in\mathbb{F}_{q^2}^{*}\), we obtain
\[
(\mu\boldsymbol{c})^q
=
\left(
(\mu a_3)^q,
(\mu a_4)^q,
(\mu a_1)^q,
(\mu a_2)^q
\right)A,
\]
where the \(q\)-th power is taken coordinatewise.

Define
$
a=\mu a_1-(\mu a_3)^q
$ and
$b=\mu a_2-(\mu a_4)^q.$
Since \(\mu^{q^2}=\mu\) and \(a_j^{q^2}=a_j\), we have
$$
-a^q
=
\mu a_3-(\mu a_1)^q
\,\, \text{and}\,\,
-b^q
=
\mu a_4-(\mu a_2)^q.
$$
Hence,
\[
\begin{aligned}
\mu\boldsymbol{c}-(\mu\boldsymbol{c})^q
&=
\bigl(
\mu a_1-(\mu a_3)^q,\,
\mu a_2-(\mu a_4)^q,
\mu a_3-(\mu a_1)^q,\,
\mu a_4-(\mu a_2)^q
\bigr)A\\
&=
(a,b,-a^q,-b^q)A
\in\mathcal{C}_2.
\end{aligned}
\]

It remains to choose \(\mu\) so that
\[
\operatorname{supp}
\left(
\mu\boldsymbol{c}-(\mu\boldsymbol{c})^q
\right)
=
\operatorname{supp}(\boldsymbol{c}).
\]
For each
\(j\in\operatorname{supp}(\boldsymbol{c})\), define
$
S_j
=
\left\{
\mu\in\mathbb{F}_{q^2}^{*}:
\mu c_j-\mu^q c_j^q=0
\right\}.
$
Since \(c_j\neq 0\), the equality
$
\mu c_j-\mu^q c_j^q=0
$
is equivalent to
$
(\mu c_j)^q=\mu c_j,
$
and then
$
\mu c_j\in\mathbb{F}_q^{*}.
$
Consequently,
$
S_j=c_j^{-1}\mathbb{F}_q^{*}
$
and
$
|S_j|=q-1.
$

Since
$
\operatorname{wt}(\boldsymbol{c})=i\leq q,
$
we obtain
$
\left|
\bigcup_{j\in\operatorname{supp}(\boldsymbol{c})}
S_j
\right|
\leq
i(q-1)
\leq
q(q-1)
<
q^2-1
=
|\mathbb{F}_{q^2}^{*}|.
$
Hence there exists
\[
\mu\in\mathbb{F}_{q^2}^{*}
\setminus
\bigcup_{j\in\operatorname{supp}(\boldsymbol{c})}S_j.
\]
For this choice of \(\mu\), whenever \(c_j\neq0\), then
$
\mu c_j-\mu^q c_j^q\neq0.
$
On the other hand, if \(c_j=0\), then clearly
$
\mu c_j-\mu^q c_j^q=0.
$
Hence,
\[
\operatorname{supp}
\left(
\mu\boldsymbol{c}-(\mu\boldsymbol{c})^q
\right)
=
\operatorname{supp}(\boldsymbol{c}).
\]
Thus every support of a weight-\(i\) codeword in
\(\mathcal{C}\) is also the support of a weight-\(i\)
codeword in \(\mathcal{C}_2\). Therefore,
\[
\mathcal{B}_i(\mathcal{C})
\subseteq
\mathcal{B}_i(\mathcal{C}_2).
\]
Combining this with (\ref{eq:ddd97}), we obtain
$
\mathcal{B}_i(\mathcal{C}_2)
=
\mathcal{B}_i(\mathcal{C})
$
for
$
i\in\{q-p^\ell,q-1,q\}.
$

By Proposition \ref{lem11}, whenever
\(\mathcal{B}_i(\mathcal{C})\neq\emptyset\),
the incidence structure
$
\left(U_{q+1},\mathcal{B}_i(\mathcal{C})\right)
$
is a \(3\)-design. Hence,
$
\left(U_{q+1},\mathcal{B}_i(\mathcal{C}_2)\right)
$
is also a \(3\)-design for
\[
i\in\{q-p^\ell,q-1,q\}.
\]

It remains to consider \(i=q+1\). If
\(\mathcal{B}_{q+1}(\mathcal{C}_2)\neq\emptyset\),
then every codeword of weight \(q+1\) has full support
\(U_{q+1}\). Thus,
$
\mathcal{B}_{q+1}(\mathcal{C}_2)
=
\{U_{q+1}\},
$
which trivially forms a
$
3-(q+1,q+1,1)
$
design.
Therefore, the codewords of every nonzero weight in
\(\mathcal{C}_2\) support a \(3\)-design.

From the above argument,
\[
\mathcal{B}_{q-p^\ell}(\mathcal{C}_2)
=
\mathcal{B}_{q-p^\ell}(\mathcal{C}).
\]
By Theorem \ref{theorem1}, the minimum weight
codewords of \(\mathcal{C}\) support a
\[
3-(q+1,q-p^\ell,\lambda)
\]
design with
\[
\lambda
=
\frac{
(q-p^\ell)(q-p^\ell-1)(q-p^\ell-2)
}{
p^{3\ell}-p^\ell
}.
\]
 This completes the proof.
\end{proof}

In the following, we give three examples to illustrate Theorem \ref{thm:C2} and Theorem \ref{lem22}. These results are consistent with those obtained by Magma.

\begin{example}
Let $p=3$, $s=1$ and $m=2$. Then $q=9$ and $\ell=\gcd(m,s)=1$. By {Theorem \ref{thm:C2}}, $\mathcal{C}_2$ is a nonlinear code with cardinality \(9^4\),
 whose possible nonzero weights are $\{10,9,8,6\}$. Moreover, by Theorem \ref{lem22}, the codewords of minimum weight $6$ in $\mathcal{C}_2$ support a $3$-$(10,6,5)$ design.
\end{example}

\begin{example}
Let $p=5$, $s=1$ and $m=2$. Then $q=25$ and $\ell=\gcd(m,s)=1$. By {Theorem \ref{thm:C2}}, $\mathcal{C}_2$ is a nonlinear code with cardinality \(25^4\),
 whose possible nonzero weights are $\{26,25,24,20\}$. Moreover, by Theorem \ref{lem22}, the codewords of minimum weight $20$ in $\mathcal{C}_2$ support a $3$-$(26,20,57)$ design.
\end{example}

\begin{example}
Let $p=3$, $s=2$ and $m=4$. Then $q=81$ and $\ell=\gcd(m,s)=2$. By {Theorem \ref{thm:C2}}, $\mathcal{C}_2$ is a nonlinear code with cardinality \(81^4\),
 whose possible nonzero weights are $\{82,81,80,72\}$. Moreover, by Theorem \ref{lem22}, the codewords of minimum weight $72$ in $\mathcal{C}_2$ support a $3$-$(82,72,497)$ design.
\end{example}

\section{Concluding remarks}

In \cite{WangTangDing}, the authors investigated the cyclic code $\mathcal{C}_1$ and showed that its minimum-weight codewords support a family of $3$-designs. In this paper, we further study two codes $\mathcal{C}$ and $\mathcal{C}_2$ associated with $\mathcal{C}_1$, where $\mathcal{C}_2$ is nonlinear and $\mathcal{C}$ is linear with
\[
\mathcal{C}=\{\mathbf{c}_1+\mathbf{c}_2:\mathbf{c}_1\in \mathcal{C}_1,\;\mathbf{c}_2\in \mathcal{C}_2\}.
\]
We proved that the codewords of every nonzero Hamming weight in $\mathcal{C}$ and $\mathcal{C}_2$ support $3$-designs. Moreover, the same property also holds for $\mathcal{C}_1$. Hence, all three codes share the common feature that the supports of the codewords of every nonzero Hamming weight form $3$-designs, yielding infinite families of $3$-designs arising from different algebraic structures.

For the linear code $\mathcal{C}$, we further determined its parameters and weight distribution, and showed that the minimum weight codewords of $\mathcal{C}^\perp$ also support $3$-designs. As applications, we constructed a $q$-ary entanglement-assisted quantum error-correcting code with parameters
\[
[[q+1,q-3,4;4]]_q,
\]
and proved that $\mathcal{C}$ is an all-symbol locally repairable code with locality $3$. In particular, when $p=3$ and $\gcd(m,s)=1$, $\mathcal{C}$ is an optimal LRC.

These results reveal a unified connection among the three codes $\mathcal{C}_1$, $\mathcal{C}_2$ and $\mathcal{C}$ through their supported $3$-designs, and further demonstrate the close interplay between coding theory, combinatorial designs, quantum error correction, and locally repairable codes.


\begin{thebibliography}{99}

\bibitem{AssmusMattson}
E. F. Assmus, Jr. and H. F. Mattson, Jr.,
``New 5-designs,''
\emph{J. Combin. Theory Ser. A}, vol. 6,
pp. 122--151, 1969.


\bibitem{Bowen}
G. Bowen,
``Entanglement required in achieving entanglement-assisted channel capacities,''
\emph{Phys. Rev. A}, vol. 66, no. 5,
Art. no. 052313, 2002.

\bibitem{BrunDevetakHsieh}
T. A. Brun, I. Devetak, and M.-H. Hsieh,
``Correcting quantum errors with entanglement,''
\emph{Science}, vol. 314, no. 5798,
pp. 436--439, Oct. 2006.

\bibitem{CadambeMazumdar}
V. R. Cadambe and A. Mazumdar,
``Bounds on the size of locally recoverable codes,''
\emph{IEEE Trans. Inf. Theory}, vol. 61, no. 11,
pp. 5787--5794, Nov. 2015.


\bibitem{ChenFangXiaFu}
B. Chen, W. Fang, S. Xia, and F. Fu,
``Constructions of optimal $(r,\delta)$ locally repairable codes via constacyclic codes,''
\emph{IEEE Trans. Commun.}, vol. 67, no. 8,
pp. 5253--5263, Aug. 2019.

\bibitem{ChenLingZhang}
B. Chen, S. Ling, and G. Zhang,
``Application of constacyclic codes to quantum MDS codes,''
\emph{IEEE Trans. Inf. Theory}, vol. 61, no. 3,
pp. 1474--1484, Mar. 2015.


\bibitem{DingTangCao}
C. Ding,
``Infinite families of 3-designs from a type of five-weight code,''
\emph{Des. Codes Cryptogr.}, vol. 86, no. 3,
pp. 703--719, 2018.

\bibitem{DingLiLi}
C. Ding and C. Li,
``Infinite families of 2-designs and 3-designs from linear codes,''
\emph{Discrete Math.}, vol. 340, no. 10,
pp. 2415--2431, Oct. 2017.

\bibitem{DingTangNMDS}
C. Ding and C. Tang,
``Infinite families of near MDS codes holding $t$-designs,''
\emph{IEEE Trans. Inf. Theory}, vol. 66, no. 9,
pp. 5419--5428, Sep. 2020.

\bibitem{DingTangSpecialPolynomials}
C. Ding and C. Tang,
``Combinatorial $t$-designs from special functions,''
\emph{Cryptogr. Commun.}, vol. 12, no. 5,
pp. 1011--1033, Sep. 2020.

\bibitem{DingTangBook}
C. Ding and C. Tang,
\emph{Designs From Linear Codes}, 2nd ed.
Singapore: World Scientific, 2022.


\bibitem{DingTangTonchevPGL}
C. Ding, C. Tang, and V. D. Tonchev,
``The projective general linear group $\mathrm{PGL}(2,2^m)$ and linear codes of length $2^m+1$,''
\emph{Des. Codes Cryptogr.}, vol. 89, no. 7,
pp. 1713--1734, 2021.

\bibitem{DuWangFan}
X. Du, R. Wang, and C. Fan,
``Infinite families of 2-designs from a class of cyclic codes,''
\emph{J. Combin. Des.}, vol. 28, no. 3,
pp. 157--170, 2020.

\bibitem{GopalanHuangSimitciYekhanin}
P. Gopalan, C. Huang, H. Simitci, and S. Yekhanin,
``On the locality of codeword symbols,''
\emph{IEEE Trans. Inf. Theory}, vol. 58, no. 11,
pp. 6925--6934, Nov. 2012.

\bibitem{Helleseth2004}
T. Helleseth, J. Lahtonen, and P. Rosendahl,
``On certain equations over finite fields and cross-correlations of $m$-sequences,''
in \emph{Coding, Cryptography and Combinatorics},
K. Feng, H. Niederreiter, and C. Xing, Eds.
Basel, Switzerland: Birkhauser, 2004,
pp. 169--176.

\bibitem{HuffmanPless}
W. C. Huffman and V. Pless,
\emph{Fundamentals of Error-Correcting Codes}.
Cambridge, U.K.: Cambridge Univ. Press, 2003.




\bibitem{KhosrovshahiLaue}
G. B. Khosrovshahi and H. Laue,
``$t$-Designs with $t\geq 3$,''
in \emph{Handbook of Combinatorial Designs}, 2nd ed.,
C. J. Colbourn and J. H. Dinitz, Eds.
Boca Raton, FL, USA: CRC Press, 2007,
pp. 79--101.

\bibitem{LiuDingMesnagerTangTonchev2}
Q. Liu, C. Ding, S. Mesnager, C. Tang, and V. D. Tonchev,
``On infinite families of narrow-sense antiprimitive BCH codes admitting 3-transitive automorphism groups and their consequences,''
\emph{IEEE Trans. Inf. Theory}, vol. 68, no. 5,
pp. 3096--3107, May 2022.

\bibitem{SunZhuWang}
Z. Sun, S. Zhu, and L. Wang,
``Optimal constacyclic locally repairable codes,''
\emph{IEEE Commun. Lett.}, vol. 23, no. 2,
pp. 206--209, Feb. 2019.



\bibitem{Tan2023}
P. Tan, C. Fan, C. Ding, C. Tang, and Z. Zhou,
``The minimum locality of linear codes,''
\emph{Des. Codes Cryptogr.}, vol. 91, no. 1,
pp. 83--114, Jan. 2023.

\bibitem{TanZhouYanParampalli}
P. Tan, Z. Zhou, H. Yan, and U. Parampalli,
``Optimal cyclic locally repairable codes via cyclotomic polynomials,''
\emph{IEEE Commun. Lett.}, vol. 23, no. 2,
pp. 202--205, Feb. 2019.



\bibitem{TangAPN}
C. Tang,
``Infinite families of 3-designs from APN functions,''
\emph{J. Combin. Des.}, vol. 28, no. 2,
pp. 97--117, Feb. 2020.

\bibitem{TangDing4Design}
C. Tang and C. Ding,
``An infinite family of linear codes supporting 4-designs,''
\emph{IEEE Trans. Inf. Theory}, vol. 67, no. 1,
pp. 244--254, Jan. 2021.

\bibitem{TangDingXiong}
C. Tang, C. Ding, and M. Xiong,
``Codes, differentially $\delta$-uniform functions, and $t$-designs,''
\emph{IEEE Trans. Inf. Theory}, vol. 66, no. 6,
pp. 3691--3703, Jun. 2020.


\bibitem{TonchevCodesDesigns}
V. D. Tonchev,
``Codes and designs,''
in \emph{Handbook of Coding Theory}, vol. 2,
V. S. Pless and W. C. Huffman, Eds.
Amsterdam, The Netherlands: Elsevier, 1998,
pp. 1229--1268.


\bibitem{WangTangDing}
X. Wang, C. Tang, and C. Ding,
``Infinite families of cyclic and negacyclic codes supporting 3-designs,''
\emph{IEEE Trans. Inf. Theory}, vol. 69, no. 4,
pp. 2341--2354, Apr. 2023.


\bibitem{WildeBrun}
M. M. Wilde and T. A. Brun,
``Optimal entanglement formulas for entanglement-assisted quantum coding,''
\emph{Phys. Rev. A}, vol. 77, no. 6,
Art. no. 064302, 2008.




\bibitem{XiangTangLiu}
C. Xiang, C. Tang, and Q. Liu,
``An infinite family of antiprimitive cyclic codes supporting Steiner systems $S(3,8,7^m+1)$,''
\emph{Des. Codes Cryptogr.}, vol. 90, no. 6,
pp. 1319--1333, Jun. 2022.

\bibitem{XuCaoQu}
G. Xu, X. Cao, and L. Qu,
``Infinite families of 3-designs and 2-designs from almost MDS codes,''
\emph{IEEE Trans. Inf. Theory}, vol. 68, no. 7,
pp. 4344--4353, Jul. 2022.


\bibitem{YanZhou}
Q. Yan and J. Zhou,
``Infinite families of linear codes supporting more $t$-designs,''
\emph{IEEE Trans. Inf. Theory}, vol. 68, no. 7,
pp. 4365--4377, Jul. 2022.






\bibitem{ZhuSunLi}
S. Zhu, Z. Sun, and P. Li,
``A class of negacyclic BCH codes and its application to quantum codes,''
\emph{Des. Codes Cryptogr.}, vol. 86, no. 10,
pp. 2139--2165, Oct. 2018.


























\end{thebibliography}
\end{document}